\documentclass[12pt]{article}
\usepackage[T1]{fontenc}
\usepackage[utf8]{inputenc}

\usepackage[round]{natbib}
\usepackage{multirow}
\usepackage[table,xcdraw]{xcolor}
\usepackage{rotating, pdflscape}
\usepackage{setspace}
\usepackage[margin=1in]{geometry}
\PassOptionsToPackage{hyphens}{url}
\usepackage[colorlinks=true, linkcolor=blue!70!black, citecolor=blue!70!black, urlcolor=blue!70!black]{hyperref}
\usepackage[stable]{footmisc}
\usepackage{subcaption}
\usepackage{amsmath}
\usepackage{amssymb,mathtools}
\usepackage{booktabs}
\usepackage{enumitem}
\usepackage{amsthm}

\newtheorem{theorem}{Theorem}
\newtheorem{lemma}{Lemma}

\theoremstyle{definition}

\theoremstyle{remark}

\newcommand{\E}{\mathbb{E}}

\newcommand{\R}{\mathbb{R}}

\title{When Does Agroforestry Income Reduce Deforestation? Evidence from a Natural Experiment in Madagascar}
\author{
  Camille DeSisto\\
  \small Sustainability Institute, Rice University, Houston, TX, USA
  \and
  Ranaivo Rasolofoson\\
  \small School of the Environment, University of Toronto, Toronto, Canada
  \and
  Michelle Foley\\
  \small MacMillan Center for International and Area Studies,\\
  \small Yale University, New Haven, CT, USA
  \and
  Harsh Parikh\textsuperscript{*}\\
  \small Yale School of Public Health, Yale University, New Haven, CT, USA\\[6pt]
  \small \textsuperscript{*}Corresponding author: \texttt{harsh.parikh@yale.edu}
}
\date{}
\begin{document}

\maketitle

\vspace{-1em}
\begin{center}
\textit{Preprint --- Not peer reviewed}\\[4pt]
\today
\end{center}
\vspace{1em}

\begin{abstract}
Tropical deforestation and rural poverty are deeply intertwined, yet isolating the causal effect of income on forest loss remains challenging. We use the 2015 global vanilla price boom—triggered by food-industry shifts toward natural flavoring—as an exogenous income shock affecting Madagascar's primary vanilla-producing region. Using a matching-augmented synthetic control design, we estimate that income gains reduced annual deforestation by 1.7 percentage points in 2017, equivalent to approximately 701 hectares of avoided forest loss. Under a monotonicity assumption linking the price boom to farmers' income, the sign of this reduced-form effect is informative about the causal direction of income on deforestation. However, effects were strongly heterogeneous: higher incomes reduced deforestation in drier, more accessible municipalities but increased clearing in wetter, low-elevation areas with high agricultural potential. These divergent patterns suggest that income simultaneously relaxes subsistence pressures driving forest dependence and raises the opportunity cost of conservation where agricultural returns are high. Our findings indicate that commodity-based agroforestry can align poverty alleviation with forest conservation under conditions of low agricultural opportunity cost. Still, policies must anticipate contexts where rising incomes amplify deforestation in agriculturally suitable land. The strategic targeting of livelihood interventions based on local agricultural potential may help reconcile development and conservation objectives in tropical forest frontiers.
\end{abstract}

\onehalfspacing

\section{Introduction}

Commodity crops sit at the center of tropical land-use change \citep{henders_trading_2015}. Since the 1980s, export-oriented agriculture has expanded rapidly as smallholder, state-enabled clearing gave way to enterprise-driven deforestation \citep{kastner_rapid_2014, rudel_changing_2009, defries_export-oriented_2013}. Converting forests to commodity production—soybeans, palm oil, cocoa, and others—drives carbon emissions and soil carbon loss, perturbs regional rainfall and land-surface temperature, and elevates infectious disease risk \citep{ordway_deforestation_2017, pendrill_agricultural_2019, van_straaten_conversion_2015, maeda_large-scale_2021, chaves_global_2020}.

Economic theory and a large empirical literature indicate that higher commodity prices increase deforestation by raising returns to land conversion in open economies \citep{angelsen_agricultural_1999, busch_what_2017, berman_crop_2023, wheeler_economic_2013, gaveau_three_2009, hargrave_economic_2013, assuncao_deforestation_2015, lundberg_maize_2022}. Yet, for rural households in low- and middle-income countries, the same crops provide crucial cash income and buffer shocks. Agroforestry—the integration of agriculture and trees/shrubs—may reconcile poverty alleviation and forest conservation goals. Commodity-based agroforestry offers relatively high earnings while retaining tree cover and ecosystem functions \citep{schroth_agroforestry_2010}; income support and alternative livelihoods can reduce reliance on risky, often illegal forest clearing \citep{ndoli_-farm_2021, emily_harwell_forests_2010, jayachandran_cash_2017, teo_reduction_2025}. The net effect of a commodity price shock on deforestation is therefore ambiguous and context-dependent: the same price increase that raises the return to land conversion also raises household income, potentially reducing pressure on forests through substitution away from more destructive livelihoods.

\begin{figure}
    \centering
    \includegraphics[width=1\linewidth]{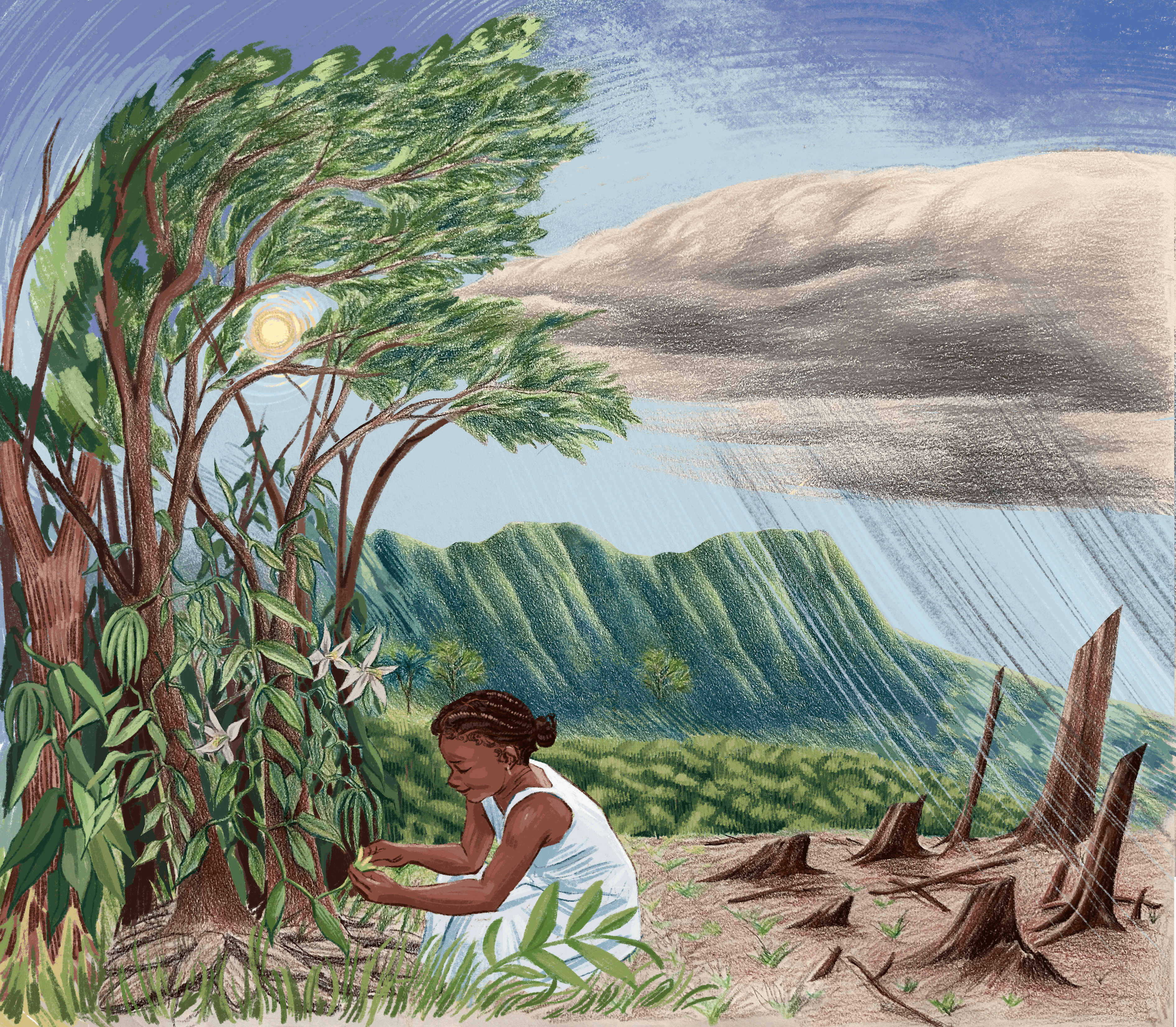}
    \caption{\textbf{Vanilla cultivation in the agroecological matrix of rural northeast Madagascar.} A farmer hand-pollinates vanilla orchids grown on tutor trees, reflecting the labor-intensive nature of the crop. Deforestation in the rural tropics is often a dangerous, last-resort activity driven by economic necessity. If higher incomes relax this pressure, wealth should reduce deforestation. However, the same income gains can incentivize forest clearing to expand agricultural production, making the net effect ambiguous. We use the 2015 global vanilla price boom---an exogenous shock to regional income in Madagascar's SAVA region---to isolate this relationship. On average, the income shock reduced deforestation. However, effects were strongly heterogeneous: income decreased deforestation in drier, more accessible municipalities and increased it in wetter, low-elevation areas with high agricultural potential.}
    \label{fig:illustration}
\end{figure}

\paragraph*{Madagascar \& Vanilla.}
Madagascar vanilla exemplifies this tension. Madagascar is a biodiversity hotspot with persistent forest loss and pervasive rural poverty \citep{myers_biodiversity_2000, vieilledent_spatial_2023, suzzi-simmons_status_2023, world_bank_world_2023}. Madagascar vanilla (\emph{Vanilla planifolia}) is a high-value export that supports hundreds of thousands of livelihoods and a substantial share of national exports \citep{andriatsitohaina_agroforestry_2024, hending_use_2018, wurz_win-win_2022}. The majority of the world's vanilla is grown in the SAVA region of northeast Madagascar, where approximately 80\% of households farm vanilla \citep{kunz_income_2020, kramer_small-scale_2023}. Although the SAVA is wealthier than many other regions in Madagascar, its residents remain burdened with high levels of poverty and food insecurity \citep{kramer_small-scale_2023, herrera_effects_2020, herrera_food_2021}.

The natural history of vanilla renders its cultivation labor-intensive and often risky, but potentially complementary to forest conservation. Vines reach maturity approximately three years after planting. Because \emph{V.~planifolia} is native to Mexico, each flower must be hand-pollinated to overcome the lack of co-occurring pollinators in Madagascar. Post-harvest curing and drying further expand the workload, and farmers often experience crop loss due to violence, pathogens, and cyclones \citep{wack_price_2023, osterhoudt_nobody_2020, iftikhar_vanilla_2023, hernandezhernandez_vanilla_2018, harison_security_2024}. Unlike other tropical commodity crops (e.g., palm oil), vanilla is an orchid best cultivated using tutor trees for shade and structural support; grown under agroforestry settings, it can complement biodiversity conservation goals \citep{hending_use_2018, wurz_win-win_2022}. In some cases, however, rainforest is cleared for vanilla cultivation, directly contributing to deforestation.

A vanilla price shock could therefore exacerbate or mitigate deforestation through multiple channels. In 2015, Madagascar vanilla experienced a price boom triggered by shifts in the food-industry towards natural flavoring \citep{noauthor_nestle_2015, press_no_2015}. We treat the 2015 vanilla price boom as a natural experiment to investigate the effects of income on deforestation. Anecdotally, following the 2015 boom, farmers rebuilt homes, bought new land, and started businesses selling furniture \citep{staevenson_bitter_2019}—activities that likely increased demand for forest-sourced planks and spurred land clearing. However, higher vanilla returns could also turn farmers away from more destructive practices: high returns may disincentivize lower-value crops (e.g., rice, often grown via swidden clearing, or \emph{tavy}) and logging in difficult-to-access forests. Although vanilla-producing areas have shown lower historical deforestation \citep{moser_economic_2008}, such associations are not necessarily causal. Using the 2015 price boom as a natural experiment, we therefore ask: \emph{how does the vanilla price boom affect deforestation in vanilla-producing municipalities, and what does this reveal about the relationship between agricultural income and forest loss?}

\paragraph*{Study Design \& Findings.}
We use the 2015 vanilla price boom as a plausibly exogenous shock, comparing deforestation trajectories in Madagascar's 73 vanilla-producing SAVA municipalities (treated) with matched non-vanilla municipalities (controls) using a matching-augmented synthetic control design. Under a monotonicity assumption linking the price boom to farmer income, the sign of our reduced-form estimates is informative about the causal direction of income on deforestation (Lemma~\ref{lem:sign}). On average, the price shock reduced deforestation by 1.7 percentage points in 2017 (${\sim}701$~ha of avoided forest loss), but effects were strongly heterogeneous: deforestation declined in drier, more accessible municipalities and increased in wetter, low-elevation areas with high agricultural potential. These patterns are consistent with income simultaneously relaxing subsistence pressures and raising the opportunity cost of conservation where agricultural returns are high.

\section{Results}

\paragraph*{Study context.}
From 2005–2012, deforestation rates across Madagascar were relatively low but cumulatively substantial; forested municipalities lost an average of 0.9\% of forest area per year during a period that coincided with low and stable vanilla prices (Fig.~\ref{fig:matching}c). 
Around 2013, the mean annual deforestation accelerated and nearly tripled to 2.6\% during 2015–2019. 
However, these associations do not necessarily disentangle the full causal relationship.
These concurrent dynamics motivate the causal design described below.

\paragraph*{Design and matching.}
We treat the 73 SAVA municipalities where vanilla is among the five main crops \citep{boone_posh_2022} as exposed to the price boom (\textit{treated}; Fig.~\ref{fig:matching}a); these municipalities produce 80–90\% of Madagascar's vanilla \citep{yoon_analysis_2020}. Municipalities with little or no vanilla production form the donor pool (\textit{controls}). Under our identification strategy, the price shock serves as an instrument-like source of variation in income: it sharply increased earnings in treated municipalities but not elsewhere, and the sign of the resulting effect is informative about the income–deforestation relationship under the monotonic relationship between vanilla prices and farmers' incomes (see Methods and Lemma~\ref{lem:sign} for formal details).

Treated municipalities tended to be at lower elevation (mean $=374$~m), steeper (mean slope $=8^{\circ}$), more sparsely populated (mean 2013–2014 density $=67$~people~km$^{-2}$), and less accessible (road density = 0.06 Km/Km$^2$) compared to controls. Treated municipalities also had higher baseline forest coverage (mean 2013–2014 forest cover $=19\%$) and a higher proportion of land covered by national protected areas (7\%) than controls.  To address these baseline differences, each treated municipality was matched to five comparable control municipalities on elevation, slope, population density, precipitation, baseline forest cover, road density, and protected area coverage before fitting the augmented synthetic control estimator (Fig.~\ref{fig:matching}b; see Methods for details).

\begin{figure}[!h]
    \centering
    \includegraphics[width=1\linewidth]{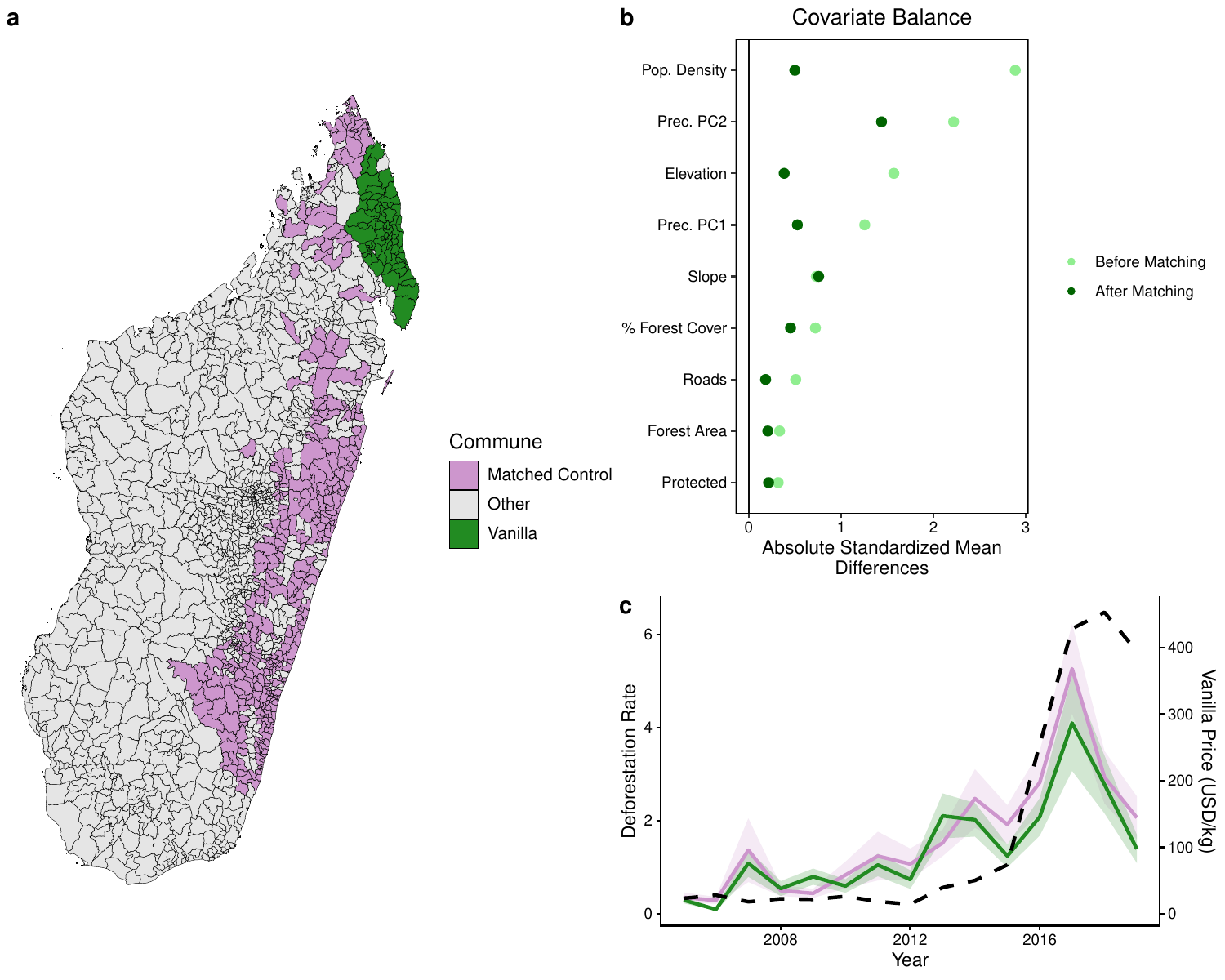}
    \caption{Statistical matching results: (a) map of municipalities across Madagascar, vanilla-farming (treated) municipalities and matched controls; (b) covariate balance before and after matching; and (c) mean deforestation in vanilla-producing municipalities (green) and matched controls (pink), alongside vanilla price (dashed black). Shaded ribbons denote 95\% confidence intervals. Note that population density, percent forest cover, and forest area represent mean values for 2013 and 2014. ``Prec.\ PC1'' and ``Prec.\ PC2'' are principal components for precipitation (Fig.~\ref{fig:pca}); ``Roads'' is the density of roads; ``Protected'' is the proportion of land area covered by a Madagascar National Park protected area.}
    \label{fig:matching}
\end{figure}

\paragraph*{Estimated effects.}
We first consider the SAVA region as one aggregated treated unit. The estimates indicate that the price shock reduced deforestation relative to matched controls (Fig.~\ref{fig:augsynth}a). In 2016, the estimated effect was $-0.7$ percentage points (pp) (95\% CI: $-2.4$ to $1.1$~pp), corresponding to ${\sim}289$~ha of avoided deforestation given 41{,}233~ha of treated forest. In 2017, the estimated reduction was $-1.7$~pp (95\% CI: $-3.5$ to $-0.1$~pp), corresponding to ${\sim}701$~ha of avoided forest loss. The negative effect was transient, persisting for ${\sim}2$ years, followed by a short-lived positive effect three years post-shock. Placebo analyses support 2015 as the appropriate shock year (Fig.~\ref{fig:shock_year}).

Modeling each municipality individually yields a similar pattern (Fig.~\ref{fig:augsynth}b): the average effect across treated units was $-0.8$~pp in 2016 and $-1.7$~pp in 2017, again followed by a short-lived positive effect. Despite this mean post-shock decrease in deforestation, effects varied widely across municipalities. Note that these effects reflect the relationship between agricultural \emph{income} and deforestation and should not be interpreted as protective effects of vanilla farming per se.

\begin{figure}[!h]
    \centering
    \includegraphics[width=1\linewidth]{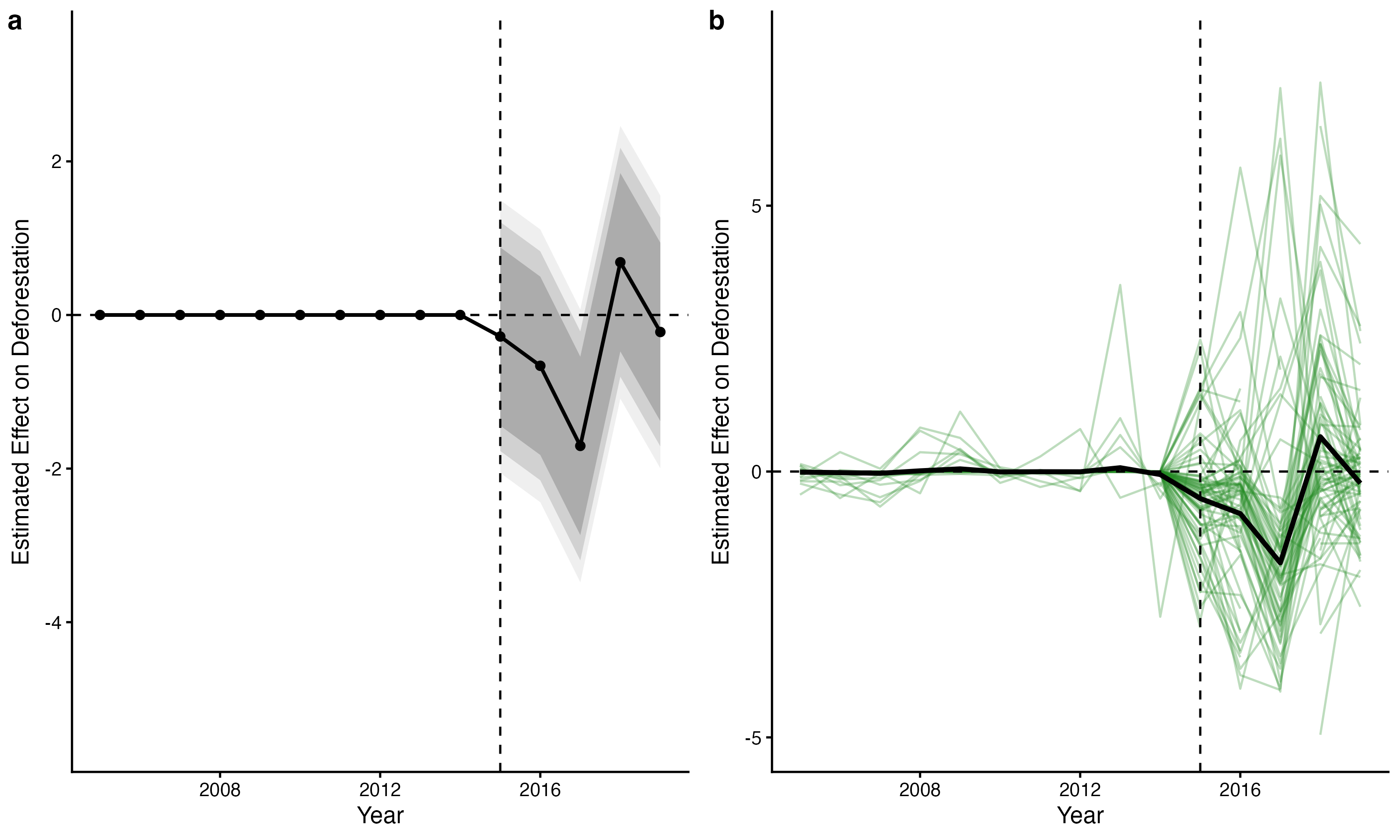}
    \caption{Augmented synthetic controls: (a) treated municipalities pooled; (b) treated municipalities modeled separately. In (a), points show estimates with 80\%, 90\%, and 95\% confidence ribbons (light to dark). In (b), green lines are municipality-level effects; the solid black line is the cross-municipality mean. Vertical dashed line marks 2015 (treatment onset). The $y$-axis in (b) is truncated to $[-5.5,\,3.5]$~pp for clarity.}
    \label{fig:augsynth}
\end{figure}

\paragraph*{Exploring heterogeneity.}
Estimated effects ranged from $-6.3$ to $+5.7$~pp in 2016 and from $-22.5$ to $+14.1$~pp in 2017. To characterize this variation, we fit a regression tree relating municipality-level effects (averaged over 2016–2017) to baseline covariates. This analysis identifies observable correlates of effect heterogeneity but does not isolate the causal moderating effect of a specific variable, as biophysical and socioeconomic variables are correlated (see Discussion). With that caveat, precipitation (PC1), elevation, and road density emerge as the primary correlates (Fig.~\ref{fig:heterogeneity}a,b). Reductions in deforestation following the price shock were concentrated in drier, more accessible municipalities. Conversely, in wetter, low-elevation municipalities with high agricultural potential, the price shock tended to increase deforestation.

\begin{figure}[!h]
    \centering
    \includegraphics[width=\linewidth]{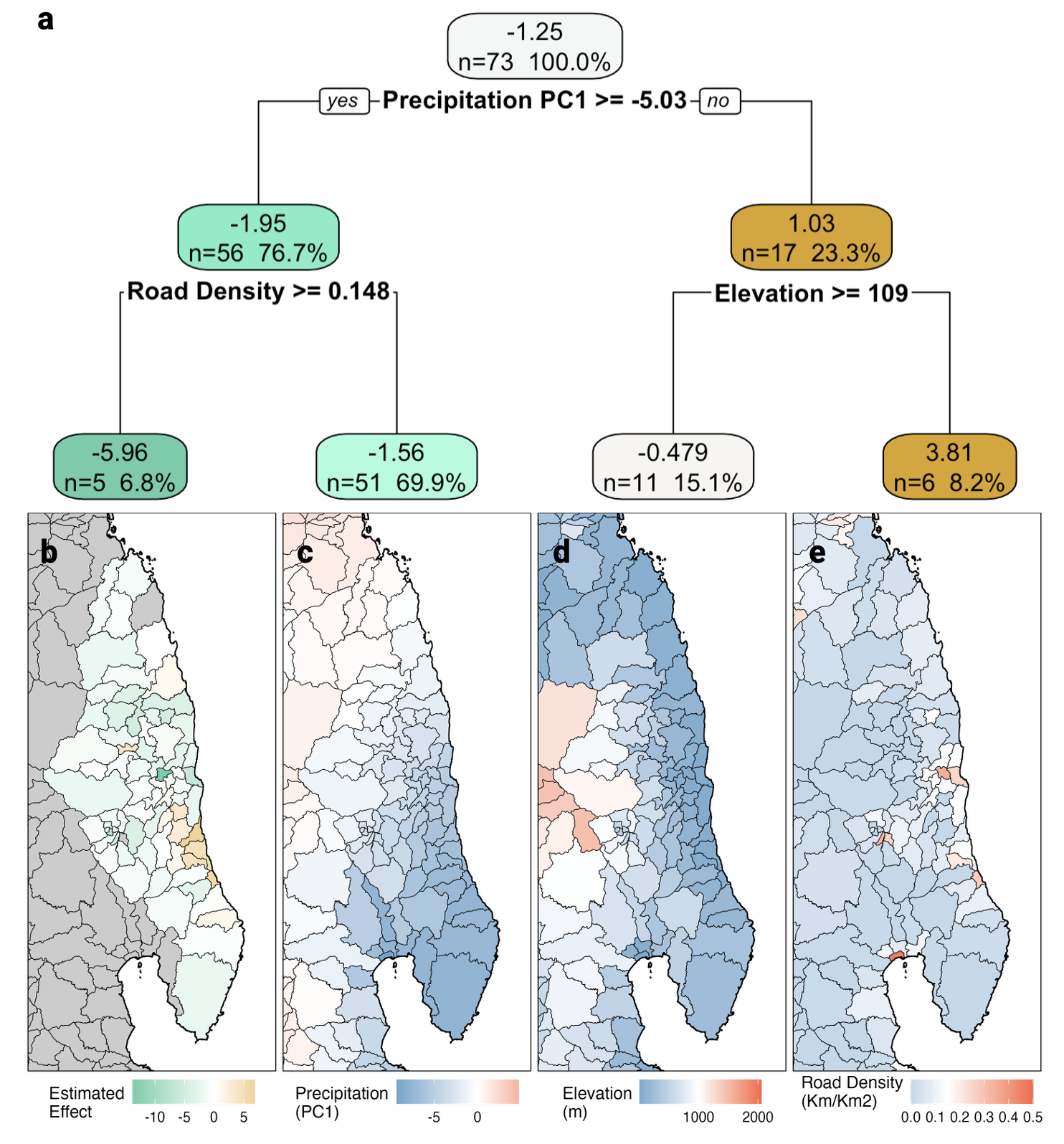}
    \caption{Heterogeneity in estimated effects of the price shock: (a) decision tree; (b) map of municipality-level effects; (c) precipitation (PC1); (d) elevation; and (e) road density. PC1 loads negatively on precipitation variables (higher values $\Rightarrow$ drier). The decision tree in panel (a) shows partitions based on municipality characteristics, where each internal node shows a splitting rule, each branch shows a decision path, and each terminal node represents a group of municipalities with shared characteristics. Nodes also denote the number of municipalities in each partition and the mean deforestation rate across those municipalities.}
    \label{fig:heterogeneity}
\end{figure}

\section{Discussion}
As evidenced by the 2015 vanilla price boom natural experiment, income was associated with a short-term decrease in deforestation in vanilla-producing municipalities. Across all treatment municipalities, the price shock decreased deforestation by an average of 1.7 percentage points in 2017 (Fig. \ref{fig:augsynth}a). Effects returned to pre-boom levels within three to four years. Importantly, the estimated effects exhibit a range of heterogeneity across municipalities (Fig. \ref{fig:augsynth}b). This effect heterogeneity is best explained by precipitation, elevation, and road density; income decreased deforestation in drier, more accessible municipalities but increased deforestation in wetter, low-elevation municipalities (Fig. \ref{fig:heterogeneity}). \textit{Income generated through commodity crop agroforestry can decrease deforestation, primarily when the baseline agricultural profitability is low.} 

\subsection{Causal Evidence}
We detected a negative effect of income on deforestation at the municipality scale. Specifically, the 2015 vanilla price boom decreased deforestation across Madagascar's vanilla-producing region (Fig. \ref{fig:augsynth}a). Increased farmer income was a likely mechanism of this effect \citep{boone_posh_2022}. Vanilla is farmed under a variety of conditions—from monoculture to diversified and from open land to integrated within forests— and does not necessarily require deforestation. When practiced sustainably, vanilla agroforestry can therefore promote landscape-level forest connectivity \citep{martin_shade-tree_2021}. Following the price shock, vanilla farmers may have substituted livelihoods that are more forest destructive for vanilla cultivation. The income generated through vanilla farming may have lowered the burdens that drive the need for forest resource extraction. For example, food security is positively associated with vanilla yield in northeast Madagascar \citep{herrera_food_2021}. Our finding that vanilla income decreased deforestation contrasts the pervasive notion that increased agricultural prices exacerbate deforestation \citep{angelsen_agricultural_1999, busch_what_2017, miranda_impacts_2024, berman_crop_2023}. Therefore, we contribute key evidence to recent discourse that, under certain circumstances, agriculture can reduce deforestation \citep{teo_reduction_2025}. 

However, the negative effects of the 2015 Madagascar vanilla boom on deforestation were short-lived. We estimated that the vanilla price boom decreased deforestation for two years, before potentially even increasing deforestation in 2018 (Fig. \ref{fig:augsynth}a). There may be a tipping point at which increased income flips from average net decreases to increases in deforestation. Note that vanilla prices peaked in 2018; there could be a price point at which vanilla income exacerbates deforestation. Additional research is required to test price tipping points for deforestation. Furthermore, when vanilla prices dropped in 2020 (after our study period), farmer income dramatically decreased  \citep{harison_compounding_2024}. This led people to reduce their food intake \citep{harison_compounding_2024}, exacerbating the already high levels of food insecurity among vanilla farmers \citep{herrera_food_2021}. Conservation and development policy must therefore consider the nuances of price effects. For example, incorporating livelihood practices such as animal husbandry and crop diversification may increase the resilience of vanilla farmers to price volatility \citep{kunz_income_2020, fleming_impact_2025}.

\subsection{Exploring Heterogeneity}
In areas of high agricultural productivity, the economic incentive to deforest dominated the conservation benefits of agroforestry. Although the vanilla price shock decreased deforestation on average, it increased deforestation under certain circumstances. Specifically, the price shock increased deforestation in high-precipitation, low-elevation municipalities (Fig. \ref{fig:augsynth}a, \ref{fig:heterogeneity}). Rainfall is associated with increased agricultural productivity in Madagascar \citep{bruelle_short-_2015, rigden_retrospective_2022}. Elevation also constrains agriculture \citep{liang_extension_2023}, and deforestation disproportionately occurs at low elevations \citep{chen_disproportionate_2024}. Therefore, when vanilla prices spiked, the opportunity cost of not farming in high-precipitation, low-elevation areas was high. In agriculturally productive areas, vanilla farmers may have therefore expanded vanilla cultivation to directly capitalize on higher pay when vanilla prices were high. A plausible indirect mechanism is that income generated from the vanilla price boom financially empowered farmers to purchase and develop new land or start other extractive businesses \citep{staevenson_bitter_2019}. The opportunity to profit likely drives observed variation in the response to price shocks at the municipality level. Because the vanilla price boom has varied effects on household wealth and health \citep{boone_posh_2022}, heterogeneity in effects on deforestation rates likely also occurs at the household level and is an important topic for future research.  

Accessibility, including via roads, is often associated with higher agricultural return, and thus higher rates of deforestation \citep{barbier_explaining_2004}. However, our heterogeneity analysis suggests that higher vanilla incomes reduced deforestation in drier and more accessible municipalities. Our heterogeneity analysis does not isolate the causal effect of a specific moderator. That the price hike reduced deforestation in more accessible areas is therefore not necessarily the moderating effect of roads but can be the effects of other moderators that correlate with roads. For example, higher road density correlates with lower precipitation. Reduced deforestation in higher road density areas may also be explained by the little forest cover left in these areas, even at baseline. 

\subsection{Limitations}
Our results are limited by data availability. Although longitudinal data on municipality-level vanilla production would strengthen our analysis, such data are not publicly available. The global data products we used may also obscure ecological nuances. For example, the limited number of weather stations in Madagascar may reduce the accuracy of WorldClim precipitation data, even though WorldClim data are commonly used in research across Madagascar (e.g., \citep{brown_predicting_2015, morelli_fate_2020}). Additionally, the global forest cover data—upon which the data used here are partially based—tend to overestimate Madagascar’s humid forests, which are the focus of this analysis \citep{rafanoharana_tree_2023}. However, our treated and control municipalities were both in eastern Madagascar's humid and subhumid bioclimatic zones, alleviating the concern about bias from incorrect estimation of forest cover. Additionally, data quality issues would only bias our results if measurement errors were systematically correlated with treatment assignment or timing of the vanilla price shock, which is unlikely. Nevertheless, increased availability of environmental data in Madagascar would strengthen our ability to construct accurate counterfactuals and estimate how income influences deforestation. 

Socio-political factors not considered in our heterogeneity analysis may shape how vanilla prices affect deforestation. For example, the 2015 Madagascar vanilla price boom was associated with dramatically heightened violence \citep{wack_price_2023, osterhoudt_nobody_2020}. Vanilla theft and vigilante justice against thieves increase the risk of farming vanilla, potentially incentivizing other types of agriculture over vanilla and affecting deforestation rates. Violence may therefore be a mechanism through which vanilla income affects deforestation. Further, political upheaval is known to disrupt tropical agricultural economies; for example, political revolutions have driven forest clearing for agriculture across the tropics (e.g., Uganda \citep{namaalwa_profitability_2001}, Rwanda \citep{kanyamibwa_impact_1998}, Colombia \citep{alvarez_forests_2003}). Our study coincides with a significant post-crisis period following the 2009 coup in Madagascar, during which governance disruptions may have independently accelerated forest loss. Recent work by Neugarten et al.\ \citep{neugarten_effect_2024} has documented a link between the post-crisis period and deforestation; annual deforestation rates accelerated during 2014–2017, particularly within community forest management areas. This post-crisis period would be a limitation of our study if there were differential effects of the event across municipalities. However, because this event affected the entire nation, we do not believe that it confounded our estimates.

\subsection{Conclusions}
Using a robust causal inference study design, we demonstrate that the 2015 vanilla price boom temporarily reduced deforestation across Madagascar’s vanilla-producing municipalities, on average (Fig. \ref{fig:augsynth}). Contrasting existing paradigms that increased open market agricultural prices exacerbate deforestation  \citep{angelsen_agricultural_1999, busch_what_2017, miranda_impacts_2024}, our results emphasize that income generated from commodity crop agriculture can slow deforestation under conditions with low agricultural opportunity cost. Critically, we demonstrate that the environmental outcomes of increased income are heterogeneous (Fig. \ref{fig:heterogeneity}). In particular, deforestation depended on climatic and topographic constraints relevant to agricultural productivity and profitability. While sustainable agroforestry should be considered in policies to alleviate economic pressures on forests, policymakers and practitioners must consider how profitability shapes economic landscapes of deforestation.  

\section{Methods}

\subsection{Study design}
We estimate the causal effect of income on deforestation by using the 2015 global vanilla price boom as an exogenous shock that differentially increased earnings in Madagascar’s main vanilla region (SAVA) relative to other municipalities. The design combines cross-sectional exposure to vanilla cultivation with time-series variation in prices, and uses a matching–augmented synthetic control estimator to adjust for place-based confounding and time-varying latent factors. We report effects for (i) the aggregated treated region and (ii) individual municipalities, then examine heterogeneity by biophysical context.

\subsection{Units, exposure, and period}
Municipalities (i.e., communes) are the unit of analysis. Treated municipalities are in SAVA (Sambava, Antalaha, Voh\'{e}mar, Andapa), where 80–90\% of Madagascar’s vanilla is produced\citep{yoon_analysis_2020}, and list vanilla among their five main crops\citep{boone_posh_2022} (\emph{n}=73). The donor pool comprises municipalities outside SAVA with little to no vanilla production. The study window is 2004–2019; 2004–2014 forms the pre-shock period, and 2015 onward marks the post-shock period (alternative shock years are assessed in sensitivity analyses).

\subsection{Outcome and covariates}
Annual forest cover ($F_{i,t}$, hectares) is derived from Vieilledent et al (2018)\citep{vieilledent_combining_2018}, built on Hansen et al (2013)\citep{hansen_high-resolution_2013}, by summing forest pixels within municipal boundaries. The annual deforestation rate is
\(
Y_{i,t} \;=\; -100 \times \frac{F_{i,t}-F_{i,t-1}}{F_{i,t-1}},
\)
so that positive values indicate increased forest loss. Baseline covariates used for design and heterogeneity analyses include population density (WorldPop), elevation and slope (SRTM via \texttt{elevatr}\citep{hollister_elevatr_2023}), long-run precipitation (WorldClim\citep{fick_worldclim_2017}), protection (proportion of area covered by a Madagascar National Park protected area), and road density (Km road/ Km$^2$ land; World Bank Group). Monthly precipitation (12 variables) is summarized using principal components; PC1 captures a wet–dry gradient (higher = drier) and PC2 captures seasonality (higher = less seasonal). Unless stated, covariates are averaged over 2013–2014.

\subsection{Statistical Analysis}
We provide a detailed formulation and discussion of our statistical analysis, identification and estimation in Appendix~\ref{app:identification}.
\subsubsection{Identification strategy}
Global price movements raise income primarily where vanilla is produced. We therefore treat the interaction of (i) municipal exposure to vanilla and (ii) the post-2015 period as the realized treatment. Note that municipalities are considered exposed to vanilla if they are in the SAVA region and if vanilla is among their top exports, and unexposed to vanilla if they are not in the SAVA region and vanilla is not among their top exports. Identification rests on: (A1) no anticipation (pre-2015 trajectories unaffected by the boom); (A2) limited spillovers to non-SAVA municipalities in the short run; (A3) shock relevance and monotonicity (the boom weakly increases income for exposed municipalities); and (A4) exclusion after adjustment: conditional on observed covariates and latent factors recovered from pre-period fit, the price boom affects deforestation only through income in exposed places. Formal notation and a reduced-form/first-stage derivation are provided in Appendix~\ref{app:identification}.

\subsubsection{Constructing a comparable donor pool}\label{sec:matching}
To improve design comparability, we restrict the donor pool via statistical matching using \texttt{MatchIt}\citep{ho_matchit_2011}. Each treated municipality is matched to five controls based on baseline percent forest cover, total forest area, population density, elevation, slope, precipitation PCs (PC1, PC2), percent protected, and road density. We audit matches to exclude implausible donors (e.g., arid municipalities unlikely to support vanilla). Balance is assessed via standardized mean differences and pre/post-matching plots (Fig.~\ref{fig:matching}b). Results are robust to the choice of K (number of controls matched to a treated unit). 

\subsubsection{Estimation: augmented synthetic control}
We implement augmented synthetic control (ASC) using \texttt{augsynth}\citep{ben-michael_augmented_2021}. For each treated municipality, ASC estimates a convex combination of matched controls that reproduces its \emph{pre-shock} deforestation path and augments the synthetic weights with outcome regression to correct residual imbalance. The difference between observed and synthetic outcomes post-shock is the estimated effect. We (i) aggregate treated municipalities by area weighting and estimate a pooled effect, and (ii) estimate municipality-level effects and summarize them across units and years. Pointwise confidence intervals follow \texttt{augsynth} defaults. See Fig. \ref{fig:augsynth_weights} for model weights. Placebo-in-time checks support 2015 as the shock year (Extended Data Fig.~\ref{fig:shock_year}).

\subsubsection{Heterogeneity analysis}\label{sec:heterogeneity}
To reduce short-run noise, we average municipality-level effects over 2016–2017 and relate these to environmental conditions used in matching. We fit a regression tree (\texttt{rpart}, ANOVA criterion) with complexity constraints (minimum split = 5; minimum bucket = 4; maximum depth = 6; complexity parameter = 0.001). Precipitation (PC1), elevation, and road density emerge as primary moderators (Fig.~\ref{fig:heterogeneity}).

\subsubsection{Sensitivity analyses}
We assess robustness to: (i) shock timing (2014, 2016, 2017 placebos); (ii) matching protocol (with/without audit; K=1 controls); and (iii) donor pool restrictions (eastern humid zones). Results are directionally consistent (Extended Data / Supplementary Figures).

\subsection{Software and reproducibility}
Analyses were conducted in \texttt{R}~4.3.1\citep{r_core_team_r_2023} with \texttt{MatchIt}\citep{ho_matchit_2011} and \texttt{augsynth}\citep{ben-michael_augmented_2021}. Preprocessing and figure code are available upon request; data sources are cited above.

\clearpage
\bibliography{references}

\clearpage
\appendix
\section{Effect Identification and Estimation}\label{app:identification}

\subsection{Setup and notation}

Let $i=1,\ldots,N$ index municipalities and $t=t_0,\ldots,t_1$ index years. We observe
\[
\{Y_{i,t},\,X_i,\,V_i,\,Z_t\}_{i,t},
\]
where $Y_{i,t}\in\R$ is the annual percent deforestation rate, $X_i\in\R^p$ is a vector of time-invariant covariates (elevation, slope, precipitation PCs, baseline forest cover, road density, protected-area share), $V_i\in\{0,1\}$ indicates exposure to vanilla ($V_i=1$ if municipality $i$ is in SAVA and vanilla is among its five main crops; $V_i=0$ otherwise), and $Z_t\in\{0,1\}$ is the post-shock indicator ($Z_t=1$ for $t\ge t^\star$ with $t^\star=2015$; $Z_t=0$ otherwise). Let $N_1:=\sum_{i=1}^{N}V_i$ and $N_0:=N-N_1$ denote the number of treated and control municipalities, respectively. The realized treatment is
\[
D_{i,t}\;:=\;V_i\,Z_t\;\in\;\{0,1\}.
\]

\subsection{Structural model and reduced form}\label{sec:structural}

We posit the existence of a latent income variable $W_{i,t}\in\R$ and time-varying unobservables $U_{i,t}\in\R^r$ (with $r\ge 1$) satisfying the following structural relations:
\begin{align}
Y_{i,t} &= \alpha(X_i)^\top U_{i,t} + \gamma(X_i)\,W_{i,t} + \varepsilon_{i,t}, \label{eq:Y-struct}\\
W_{i,t} &= \rho(X_i)^\top U_{i,t} + \delta(X_i)\,V_i\,Z_t + \nu_{i,t}, \label{eq:W-struct}
\end{align}
where $\alpha:\R^p\to\R^r$, $\rho:\R^p\to\R^r$, and $\gamma,\delta:\R^p\to\R$ are measurable functions, and the errors satisfy
\[
\E[\varepsilon_{i,t}\mid X_i, U_{i,t}, W_{i,t}] = 0, \qquad \E[\nu_{i,t}\mid X_i, U_{i,t}, V_i, Z_t] = 0.
\]

Equation~\eqref{eq:Y-struct} models deforestation as a function of unobserved municipality--time confounders $U_{i,t}$, latent income $W_{i,t}$, and idiosyncratic shocks. The coefficient $\gamma(X_i)$ is the structural \emph{income effect}: the causal effect of a unit increase in income on deforestation, conditional on covariates. Equation~\eqref{eq:W-struct} is a first-stage equation: income depends on the same unobservables $U_{i,t}$ and is shifted by the interaction $V_i Z_t$---the exogenous income shock from the vanilla price boom. The coefficient $\delta(X_i)\ge 0$ captures the first-stage relevance of the price boom for income.

\paragraph{Remark on exposure.} Exposure $V_i$ is time-invariant and predetermined (it reflects the long-standing geography of vanilla cultivation, not a post-2015 decision). We do not model a structural equation for $V_i$; instead, we allow $V_i$ to be arbitrarily correlated with the latent factors $U_{i,t}$ through the factor loadings, which the synthetic control design absorbs.

\paragraph{Reduced form.} Substituting \eqref{eq:W-struct} into \eqref{eq:Y-struct} yields
\begin{equation}\label{eq:reduced-form}
Y_{i,t} \;=\; \underbrace{\bigl[\alpha(X_i) + \gamma(X_i)\,\rho(X_i)\bigr]^\top U_{i,t}}_{\text{latent time--place component}} \;+\; \underbrace{\gamma(X_i)\,\delta(X_i)}_{=:\;\tau(X_i)}\;V_i\,Z_t \;+\; \underbrace{\gamma(X_i)\,\nu_{i,t} + \varepsilon_{i,t}}_{=:\;\eta_{i,t}},
\end{equation}
where $\E[\eta_{i,t}\mid X_i, U_{i,t}]=0$. The reduced-form treatment effect on deforestation is $\tau(X_i) = \gamma(X_i)\,\delta(X_i)$, the product of the income effect and the first-stage relevance.

\paragraph{Interactive fixed-effects representation.} Define the composite factor loading $\Phi_i := [\alpha(X_i)+\gamma(X_i)\rho(X_i)]\in\R^r$ and let $\Lambda_t$ denote the common factors underlying $U_{i,t}$ (with municipality-specific loadings absorbed into $\Phi_i$). Then \eqref{eq:reduced-form} can be written as
\begin{equation}\label{eq:IFE}
Y_{i,t} \;=\; \Lambda_t^\top \Phi_i \;+\; \tau(X_i)\,D_{i,t} \;+\; \eta_{i,t}, \qquad \E[\eta_{i,t}\mid \Lambda_t,\Phi_i,X_i]=0.
\end{equation}
This is a standard interactive fixed-effects (IFE) panel model with heterogeneous treatment effects \citep{xu_generalized_2017,ben-michael_augmented_2021}, augmented by the observation that the treatment effect $\tau(X_i)$ has a signed structural interpretation via the decomposition $\tau = \gamma\cdot\delta$.

\subsection{Potential outcomes and estimands}

Let $Y_{i,t}(d)$ denote the potential outcome under $D_{i,t}=d\in\{0,1\}$, so the realized outcome satisfies $Y_{i,t} = Y_{i,t}(D_{i,t})$. We target the following estimands:
\begin{align}
\tau_t &:= \E\!\bigl[Y_{i,t}(1) - Y_{i,t}(0) \mid V_i=1\bigr], \qquad t\ge t^\star, \label{eq:ATT-time}\\
\tau(x) &:= \E\!\bigl[Y_{i,t}(1) - Y_{i,t}(0) \mid X_i=x,\;V_i=1\bigr], \qquad t\ge t^\star. \label{eq:CATE-x}
\end{align}
The first is the average treatment effect on the treated (ATT) at time $t$; the second is the conditional ATT (CATT), which we use in the heterogeneity analysis (Section~\ref{sec:heterogeneity}). Under \eqref{eq:IFE}, $\tau_t = \E[\tau(X_i)\mid V_i=1]$ and $\tau(x) = \gamma(x)\,\delta(x)$.

\subsection{Identifying assumptions}

We impose the following conditions, adapted to a shock--exposure setting with interactive fixed effects.

\begin{enumerate}[leftmargin=1.6em,label=\textbf{(A\arabic*)}]

\item \textbf{No anticipation.} For all $t < t^\star$, $Y_{i,t}(1) = Y_{i,t}(0)$ almost surely. \\
\emph{Content:} Prior to the 2015 price boom, potential outcomes do not depend on future treatment status. This is standard in the synthetic control literature and is empirically supported by the absence of pre-treatment trends in estimated effects (Figs.~\ref{fig:augsynth} and \ref{fig:shock_year}).

\item \textbf{No interference (SUTVA).} For all $i$ with $V_i=0$ and all $t$, the price boom does not affect municipality $i$'s deforestation rate. For treated municipalities, potential outcomes depend only on own treatment status, not on other municipalities' treatment.\\
\emph{Content:} Short-run spillovers from SAVA to non-SAVA municipalities are negligible. This is plausible because vanilla production is geographically concentrated, and the income shock did not materially affect the economic conditions of distant non-vanilla municipalities within our study window.

\item \textbf{Relevance and monotonicity.} For all $x$ in the support of $X_i$ among treated municipalities, $\delta(x) \ge 0$, with $\E[\delta(X_i) \mid V_i=1] > 0$. For municipalities with $V_i=0$, $\delta(x)=0$.\\
\emph{Content:} The price boom weakly increases income in all vanilla-producing municipalities (with a strict increase on average) and does not increase income in unexposed municipalities. The first part follows from the mechanical link between vanilla prices and farmer revenue in a region where $\sim$80\% of households farm vanilla; the second follows from the geographic concentration of vanilla production.

\item \textbf{Exclusion after conditioning.} Conditional on $(\Lambda_t, \Phi_i, X_i)$, the post-shock indicator $Z_t$ affects $Y_{i,t}$ only through income $W_{i,t}$ in exposed municipalities.\\
\emph{Content:} After accounting for latent factors and observed covariates, the price boom has no direct effect on deforestation except through its effect on income. Potential violations include the direct use of vanilla vines as a land-use practice; our reduced-form effect $\tau$ should therefore be interpreted as the total effect mediated by income broadly construed (including income-financed land expansion), rather than a pure substitution effect.

\item \textbf{Overlap and factor model regularity.}
\begin{enumerate}[label=(\roman*)]
\item For each treated unit $i$, there exist convex weights $\{w_{ij}\}_{j\in\mathcal{J}(i)}$ with $w_{ij}\ge 0$ and $\sum_j w_{ij}=1$ such that
\[
\Bigl\|\Phi_i - \sum_{j\in\mathcal{J}(i)} w_{ij}\Phi_j\Bigr\| + \Bigl\|X_i - \sum_{j\in\mathcal{J}(i)} w_{ij}X_j\Bigr\| \;=\; o_p(1).
\]
\item The number of factors $r$ is fixed and finite.
\item $\|\Lambda_t\|$ is uniformly bounded, and $\frac{1}{T_0}\sum_{t\in\mathcal{T}_0}\Lambda_t\Lambda_t^\top$ converges to a positive-definite matrix as $T_0\to\infty$.
\end{enumerate}
\emph{Content:} Part (i) requires that each treated municipality's factor loadings lie approximately in the convex hull of its matched donors'---a condition that statistical matching (Section~\ref{sec:matching}) is designed to achieve. Parts (ii)--(iii) are standard regularity conditions for IFE models \citep{xu_generalized_2017,ben-michael_augmented_2021}.

\end{enumerate}

\subsection{Identification results}

\begin{theorem}[Identification of the reduced-form ATT]\label{thm:id}
Under \textbf{(A1)--(A5)}, for any $t\ge t^\star$, the ATT is identified as
\[
\tau_t \;=\; \E\!\bigl[Y_{i,t} - Y_{i,t}^{(0)} \mid V_i=1\bigr],
\]
where $Y_{i,t}^{(0)} := \Lambda_t^\top\Phi_i + \eta_{i,t}$ is the counterfactual outcome that treated unit $i$ would have experienced absent the shock, constructed from the pre-shock fit.
\end{theorem}

\begin{proof}
By \textbf{(A1)}, for $t<t^\star$, $D_{i,t}=0$ for all $i$, so $Y_{i,t} = \Lambda_t^\top\Phi_i + \eta_{i,t}$. For $t\ge t^\star$ and $V_i=1$, equation~\eqref{eq:IFE} gives $Y_{i,t} = \Lambda_t^\top\Phi_i + \tau(X_i) + \eta_{i,t}$. Thus
\[
Y_{i,t} - Y_{i,t}^{(0)} = \tau(X_i), \qquad \text{for } V_i=1,\; t\ge t^\star.
\]
Averaging over treated units yields $\tau_t = \E[\tau(X_i)\mid V_i=1]$. The counterfactual $Y_{i,t}^{(0)}$ is not directly observed but is estimable from the pre-shock period: by \textbf{(A5)}, there exist donor weights that approximate $\Phi_i$, and the latent factor structure $\Lambda_t^\top\Phi_i$ is recoverable from the pre-shock path of outcomes. The augmented synthetic control estimator (Section~\ref{sec:ASC} below) provides a consistent estimate $\widehat{Y}_{i,t}^{(0)}$ of this counterfactual.
\end{proof}

\begin{lemma}[Sign identification for the income effect]\label{lem:sign}
Under \textbf{(A3)}, for almost every $x$ in the support of $X_i$ among treated units:
\[
\operatorname{sign}\{\tau(x)\} \;=\; \operatorname{sign}\{\gamma(x)\}.
\]
That is, the sign of the reduced-form effect $\tau(x)=\gamma(x)\,\delta(x)$ reveals the sign of the structural income effect $\gamma(x)$.
\end{lemma}

\begin{proof}
By \textbf{(A3)}, $\delta(x)\ge 0$ for all $x$ in the treated support, with strict inequality holding generically (i.e., the set $\{x:\delta(x)=0\}$ has measure zero under the treated covariate distribution, since $\E[\delta(X_i)\mid V_i=1]>0$). Where $\delta(x)>0$, we have $\operatorname{sign}\{\gamma(x)\delta(x)\} = \operatorname{sign}\{\gamma(x)\}\cdot\operatorname{sign}\{\delta(x)\} = \operatorname{sign}\{\gamma(x)\}$. Where $\delta(x)=0$, $\tau(x)=0$ and the sign is uninformative, but this occurs on a set of measure zero.
\end{proof}

\paragraph{Interpretation.} Lemma~\ref{lem:sign} establishes that the sign of our reduced-form estimates is informative about the causal direction of income on deforestation. If $\tau(x)<0$---i.e., the price boom reduced deforestation in municipalities with covariates $x$---then $\gamma(x)<0$, meaning income causally reduces deforestation for those municipalities. Conversely, $\tau(x)>0$ implies $\gamma(x)>0$: income causally increases deforestation. This sign-identification result is the key link between our reduced-form estimates and the structural question of interest.

\subsection{Augmented synthetic control: definition and estimator}\label{sec:ASC}

We implement the augmented synthetic control (ASC) method of Ben-Michael et al. (2021) \cite{ben-michael_augmented_2021}. For each treated municipality $i$ with matched donor set $\mathcal{J}(i)\subset\{j:V_j=0\}$ (obtained via statistical matching; Section~\ref{sec:matching}), let
\[
\boldsymbol{\omega}_i = (\omega_{ij})_{j\in\mathcal{J}(i)}, \qquad \omega_{ij}\ge 0,\quad \sum_{j\in\mathcal{J}(i)}\omega_{ij}=1,
\]
be convex weights. Let $\mathcal{T}_0:=\{t_0,\ldots,t^\star-1\}$ and $\mathcal{T}_1:=\{t^\star,\ldots,t_1\}$ denote pre- and post-shock periods, with $T_0:=|\mathcal{T}_0|$ and $T_1:=|\mathcal{T}_1|$.

\paragraph{Step 1: Synthetic control weights.} For each treated unit $i$, choose $\boldsymbol{\omega}_i$ to minimize the pre-shock root mean squared prediction error (RMSPE) subject to approximate covariate balance:
\begin{equation}\label{eq:sc-weights}
\min_{\boldsymbol{\omega}_i}\;\sum_{t\in\mathcal{T}_0}\Bigl(Y_{i,t} - \sum_{j\in\mathcal{J}(i)}\omega_{ij}\,Y_{j,t}\Bigr)^{\!2}
\quad\text{s.t.}\quad \bigl\|X_i - \textstyle\sum_{j}\omega_{ij}\,X_j\bigr\| \le \eta,\;\;\omega_{ij}\ge 0,\;\;\textstyle\sum_j\omega_{ij}=1,
\end{equation}
with tolerance $\eta\ge 0$. When matching has already produced near-exact covariate balance, the constraint is approximately slack.

\paragraph{Step 2: Outcome-model augmentation.} For each $t\in\mathcal{T}_0$, fit a ridge regression on the donor units:
\begin{equation}\label{eq:ridge}
\widehat{m}_t \;=\; \arg\min_{m\in\mathcal{M}}\;\sum_{j\in\mathcal{J}(i)}\bigl(Y_{j,t} - m(X_j)\bigr)^2 \;+\; \lambda\,\|m\|_{\mathcal{M}}^2,
\end{equation}
where $\mathcal{M}$ is a low-complexity function class (we use ridge regression on $X$ with a small number of pre-period outcome lags) and $\lambda\ge 0$ is a regularization parameter. The \emph{augmented} counterfactual for unit $i$ at time $t$ is
\begin{equation}\label{eq:augmented-counterfactual}
\widehat{Y}_{i,t}^{(0)} \;=\; \underbrace{\sum_{j\in\mathcal{J}(i)}\omega_{ij}\,Y_{j,t}}_{\text{synthetic control}} \;+\; \underbrace{\widehat{m}_t(X_i) - \sum_{j\in\mathcal{J}(i)}\omega_{ij}\,\widehat{m}_t(X_j)}_{\text{bias correction}}.
\end{equation}
The first term reproduces the treated unit's latent-factor trajectory via donor weighting; the second corrects for residual covariate imbalance. When the synthetic control weights achieve exact pre-period fit, the bias correction is zero; when they do not, the ridge augmentation de-biases the estimate by adjusting for the covariate mismatch \citep{ben-michael_augmented_2021}.

\paragraph{Step 3: Effect estimation and aggregation.} For each treated unit $i$ and post-shock year $t\in\mathcal{T}_1$, the estimated unit--time effect is
\begin{equation}\label{eq:unit-effect}
\widehat{\tau}_{i,t} \;=\; Y_{i,t} \;-\; \widehat{Y}_{i,t}^{(0)}.
\end{equation}
We report two summaries:
\begin{enumerate}[label=(\roman*)]
\item \textbf{Pooled effect.} Aggregate all treated municipalities into a single unit by area-weighting their outcomes before fitting \eqref{eq:sc-weights}--\eqref{eq:augmented-counterfactual}. The resulting estimate $\widehat{\tau}_t^{\text{pool}}$ is reported with pointwise confidence intervals in Fig.~\ref{fig:augsynth}a.
\item \textbf{Municipality-level effects.} Fit \eqref{eq:sc-weights}--\eqref{eq:augmented-counterfactual} separately for each treated municipality, yielding $\{\widehat{\tau}_{i,t}\}_{i:V_i=1}$. The cross-municipality average, $\bar{\tau}_t := N_1^{-1}\sum_{i:V_i=1}\widehat{\tau}_{i,t}$, is reported alongside individual trajectories in Fig.~\ref{fig:augsynth}b.
\end{enumerate}

\paragraph{Consistency.} Under the IFE model~\eqref{eq:IFE} and assumptions \textbf{(A1)--(A5)}, with $T_0\to\infty$ and $|\mathcal{J}(i)|$ fixed, Ben-Michael et al. (2021) \cite{ben-michael_augmented_2021} show that the ASC estimator satisfies
\[
\bigl|\widehat{\tau}_{i,t} - \tau(X_i)\bigr| \;=\; O_p\!\Bigl(\frac{1}{\sqrt{T_0}}\Bigr) + \text{bias from pre-treatment imbalance},
\]
where the bias term is controlled by both the synthetic control fit and the ridge augmentation. In our application, $T_0=11$ pre-shock years (2004--2014) and each treated unit is matched to $K=5$ donors, yielding donor pools of modest size. The augmentation is therefore particularly valuable: it corrects residual imbalance that pure reweighting cannot eliminate with few donors.

\paragraph{Inference.} Pointwise confidence intervals follow the default procedure in \texttt{augsynth} \citep{ben-michael_augmented_2021}, which uses a Jackknife+ approach over the pre-treatment periods. For the pooled estimator, these intervals account for estimation uncertainty in both the synthetic control weights and the ridge augmentation.

\paragraph{Remarks.}
\begin{enumerate}[label=(\roman*)]
\item When matching already delivers near-exact balance on $X$, the augmentation term in \eqref{eq:augmented-counterfactual} is small but provides a finite-sample stability guarantee.
\item The function class $\mathcal{M}$ is deliberately low-dimensional (ridge on $X$ plus a small number of lagged outcomes) to avoid extrapolation beyond the convex hull of the donor pool.
\item The notation $\boldsymbol{\omega}_i$ for synthetic control weights is distinct from $W_{i,t}$ (latent income) throughout.
\end{enumerate}

\clearpage

\section{Supplementary Information}

\subsection{Principal Component Analysis}
\renewcommand{\thefigure}{S1}
\begin{figure}[!hbt]
    \centering
    \includegraphics[width=1\linewidth]{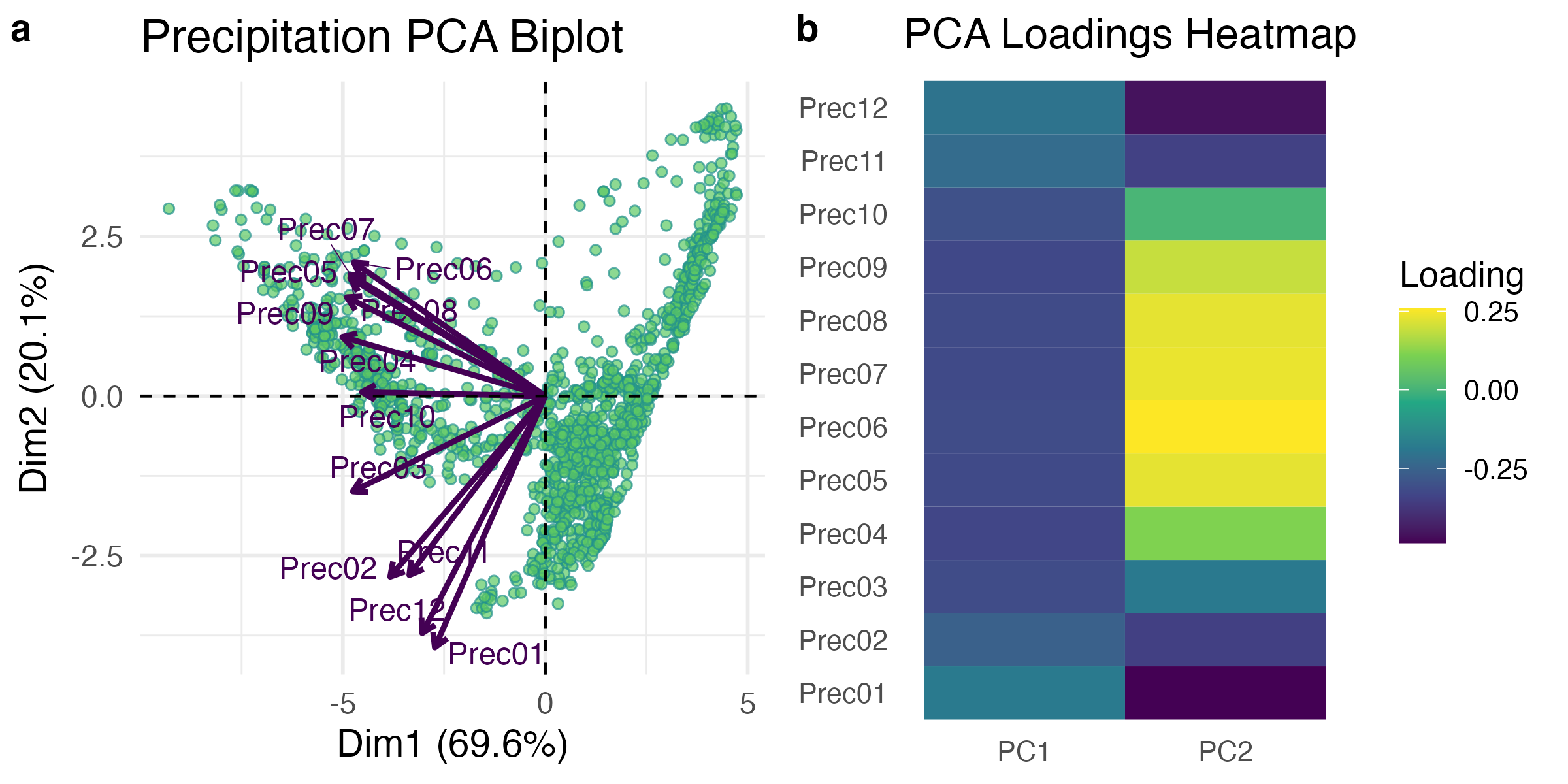}
    \caption{PCA results of historical precipitation across all municipalities: a) biplot of the first two principal components and b) a heatmap showing the loadings of the first two principal components.}
    \label{fig:pca}
\end{figure} 
\clearpage

\subsection{Robustness to Donor Pool Restrictions}

\renewcommand{\thefigure}{S2}
\begin{figure}[!hbt]
    \centering
    \includegraphics[width=1\linewidth]{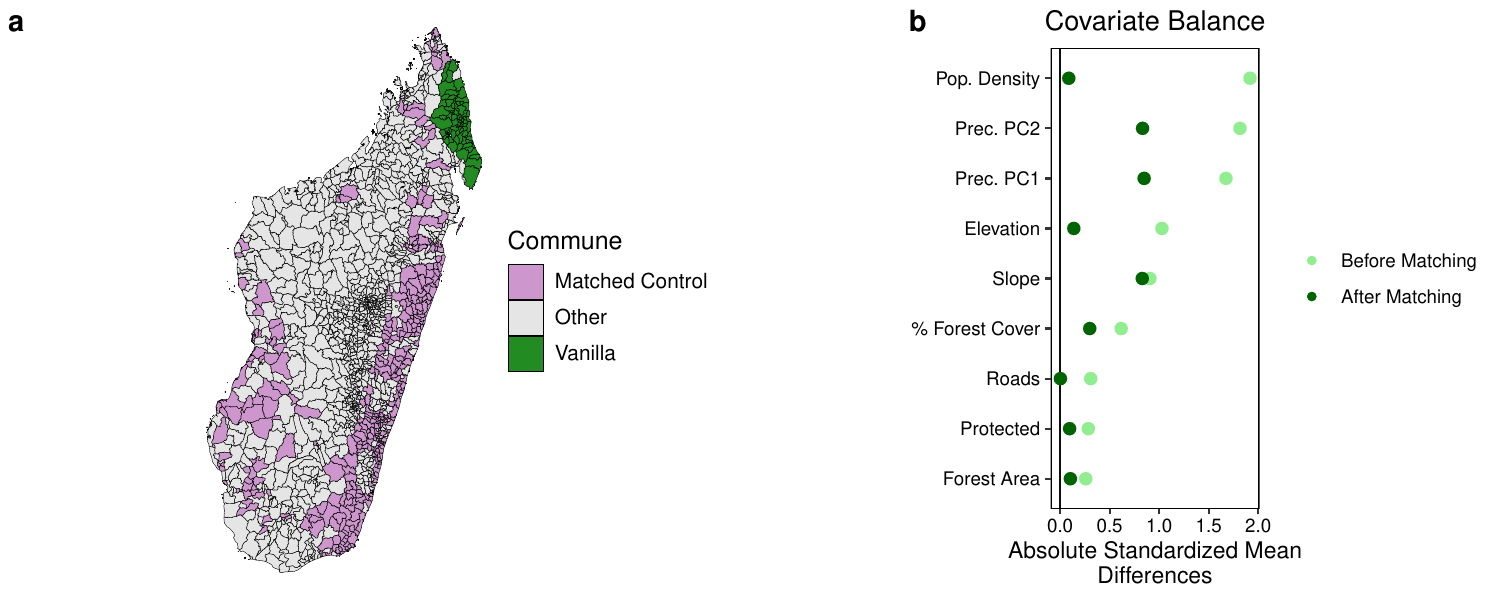}
    \caption{Statistical matching results that do not restrict the donor pool to exclude dry areas, showing: a) map of vanilla farming municipalities and matched controls and b) covariate balance pre- and post-matching. }
    \label{fig:matching_unrefined}
\end{figure} 

\renewcommand{\thefigure}{S3}
\begin{figure}[!hbt]
    \centering
    \includegraphics[width=0.8\linewidth]{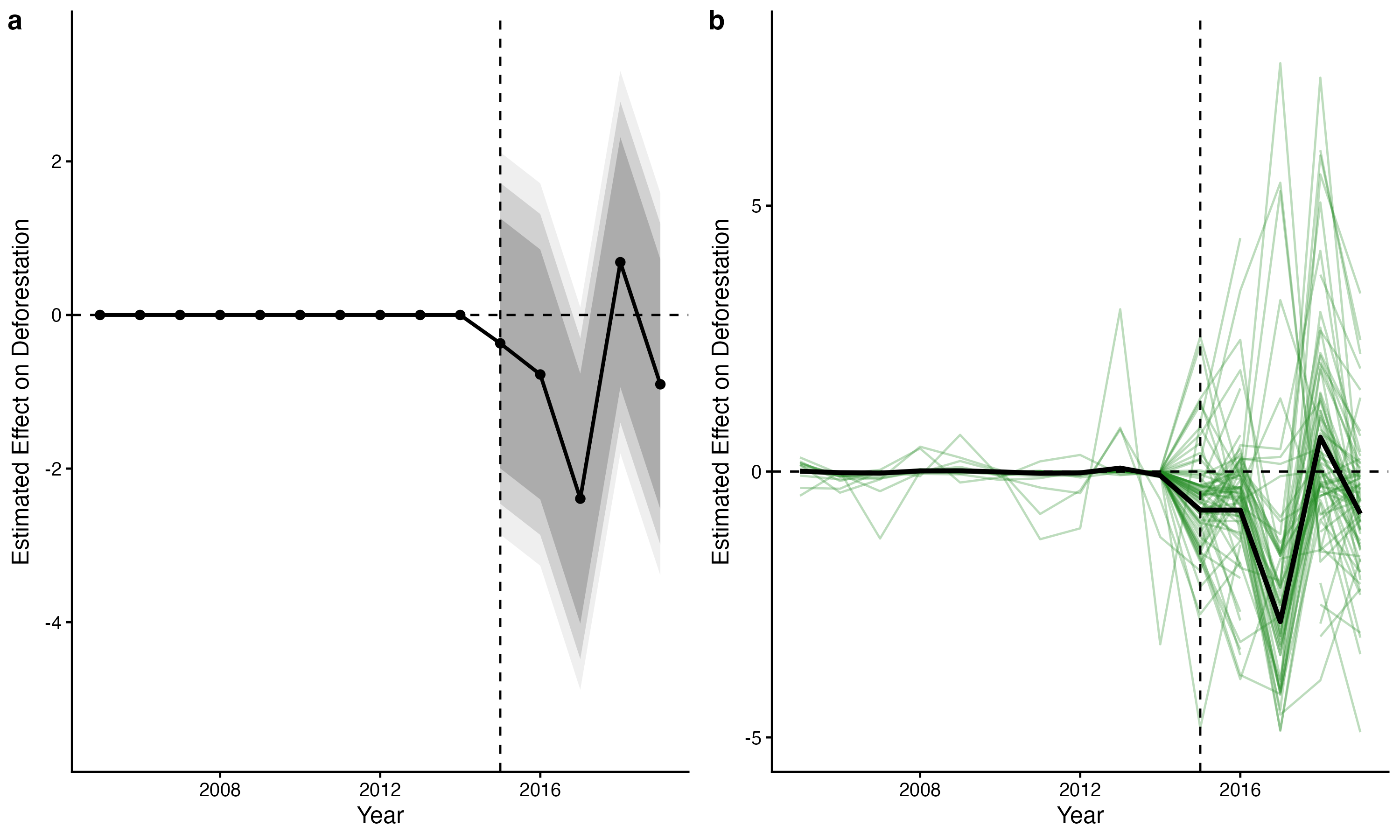}
    \caption{Augmented synthetic controls results based on the matched set where the donor pool was not restricted: a) the vanilla-producing municipalities considered together as one region and b) the vanilla-producing municipalities modeled separately. In panel (a), black dots represent estimates, and ribbons represent confidence intervals (light gray = 80\%, medium gray = 90\%, dark gray = 95\%). In panel (b), each green line represents a vanilla-producing municipality, and the solid black line represents the mean. Dashed lines indicate the treatment year (2015). Note that the y-axis is limited to -5.5 to 3.5 to improve visualization.}
    \label{fig:augsynth_unrefined}
\end{figure} 

\clearpage

\subsection{Sensitivity to the choice of Shock Year}

\renewcommand{\thefigure}{S4}
\begin{figure}[!hbt]
    \centering
    \includegraphics[width=1\linewidth]{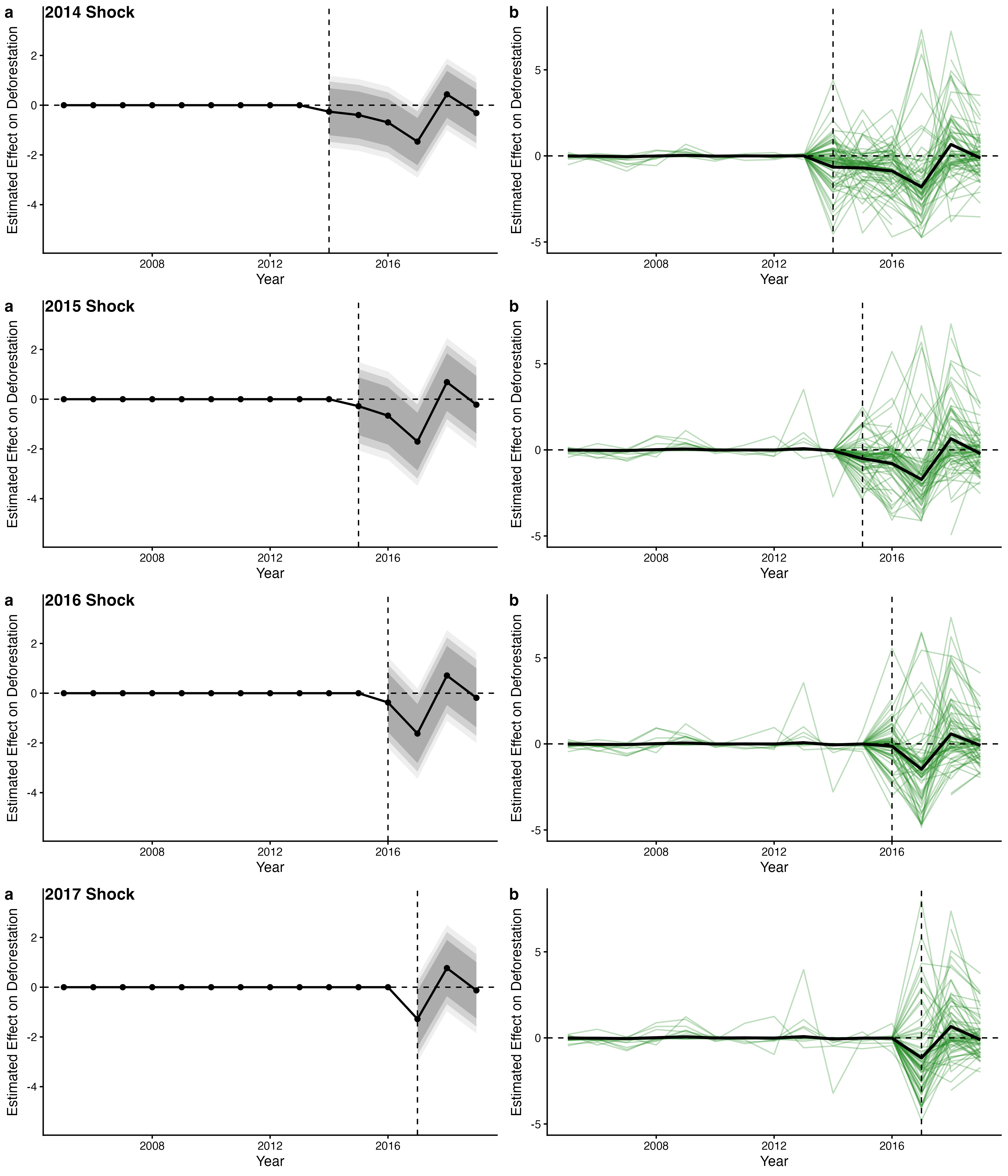}
    \caption{Augmented synthetic controls results for 2014, 2015, 2016, and 2017 for \textbf{K = 5}.}
    \label{fig:shock_year}
\end{figure} 

\clearpage

\subsection{Robustness to Number of Matches (K) in the Statistical Matching}

\renewcommand{\thefigure}{S5}
\begin{figure}[!hbt]
    \centering
    \includegraphics[width=1\linewidth]{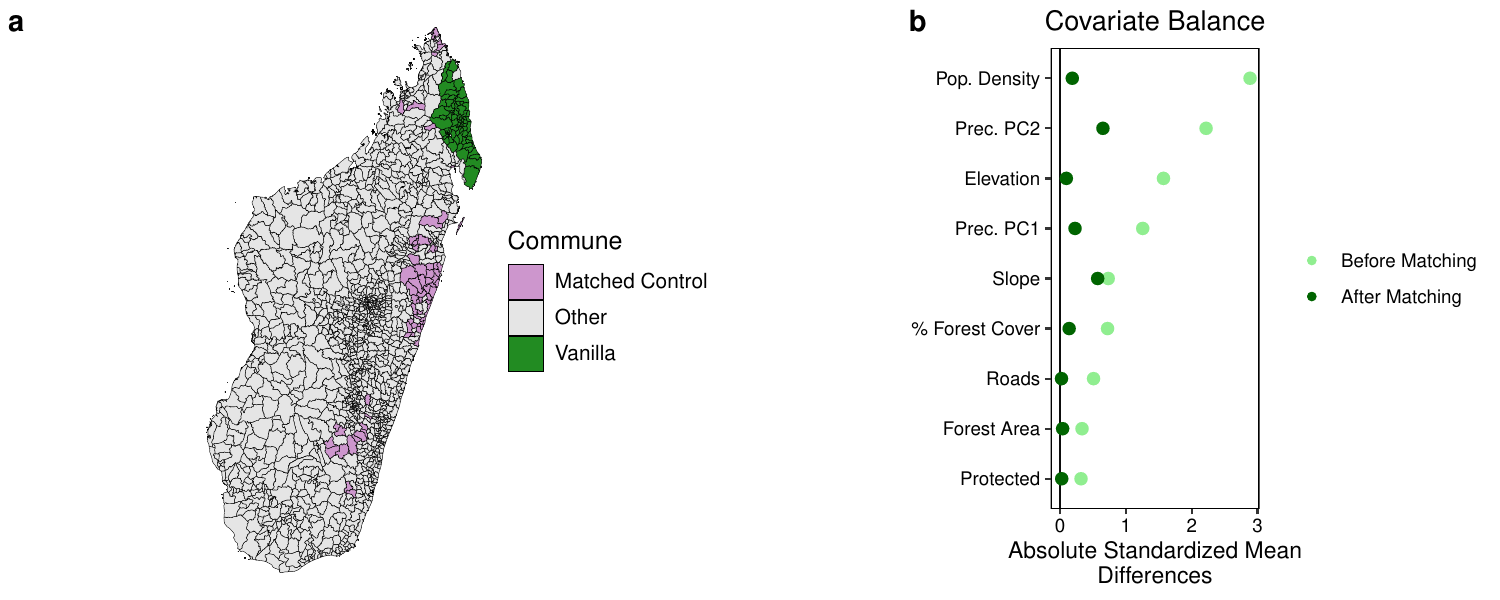}
    \caption{Statistical matching results for \textbf{K = 1}: (a) map of vanilla-farming municipalities and matched controls; (b) covariate balance before and after matching; and (c) mean deforestation in vanilla-producing municipalities (green) and matched controls (pink), alongside vanilla price (dashed black). Shaded ribbons denote 95\% confidence intervals.}
    \label{fig:matching_k1}
\end{figure}

\renewcommand{\thefigure}{S6}
\begin{figure}[!hbt]
    \centering
    \includegraphics[width=0.8\linewidth]{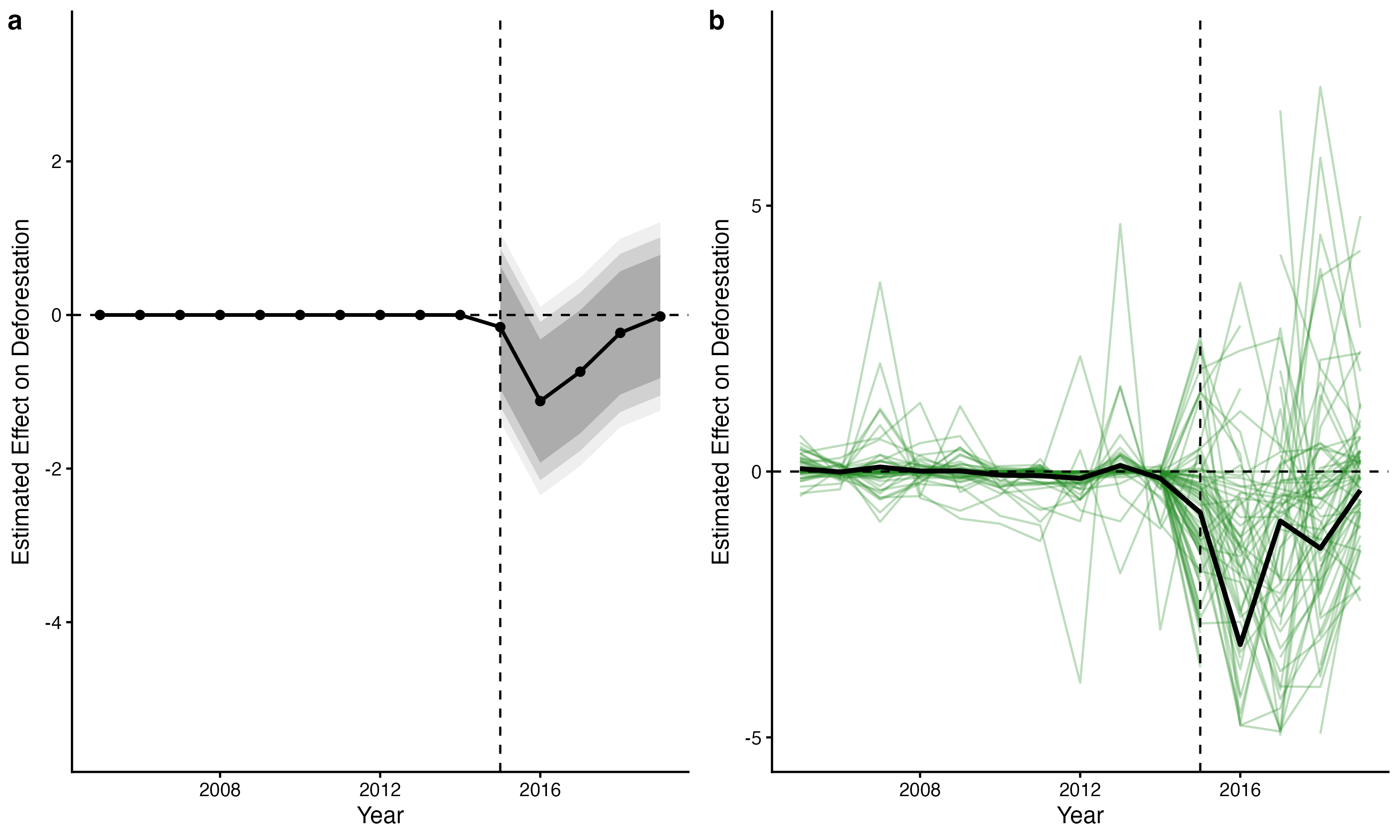}
    \caption{Augmented synthetic controls for \textbf{K = 1}: (a) treated municipalities pooled; (b) treated municipalities modeled separately. In (a), points show estimates with 80\%, 90\%, and 95\% confidence ribbons (light to dark). In (b), green lines are municipality-level effects; the solid black line is the cross-municipality mean. Vertical dashed line marks 2015 (treatment onset). The $y$-axis is truncated to $[-5.5,\,3.5]$~pp for clarity.}
    \label{fig:augsynth_k1}
\end{figure}

\subsection{Augmented Synthetic Control Model Weight Distribution}

\renewcommand{\thefigure}{S7}
\begin{figure}[!hbt]
    \centering
    \includegraphics[width=1\linewidth]{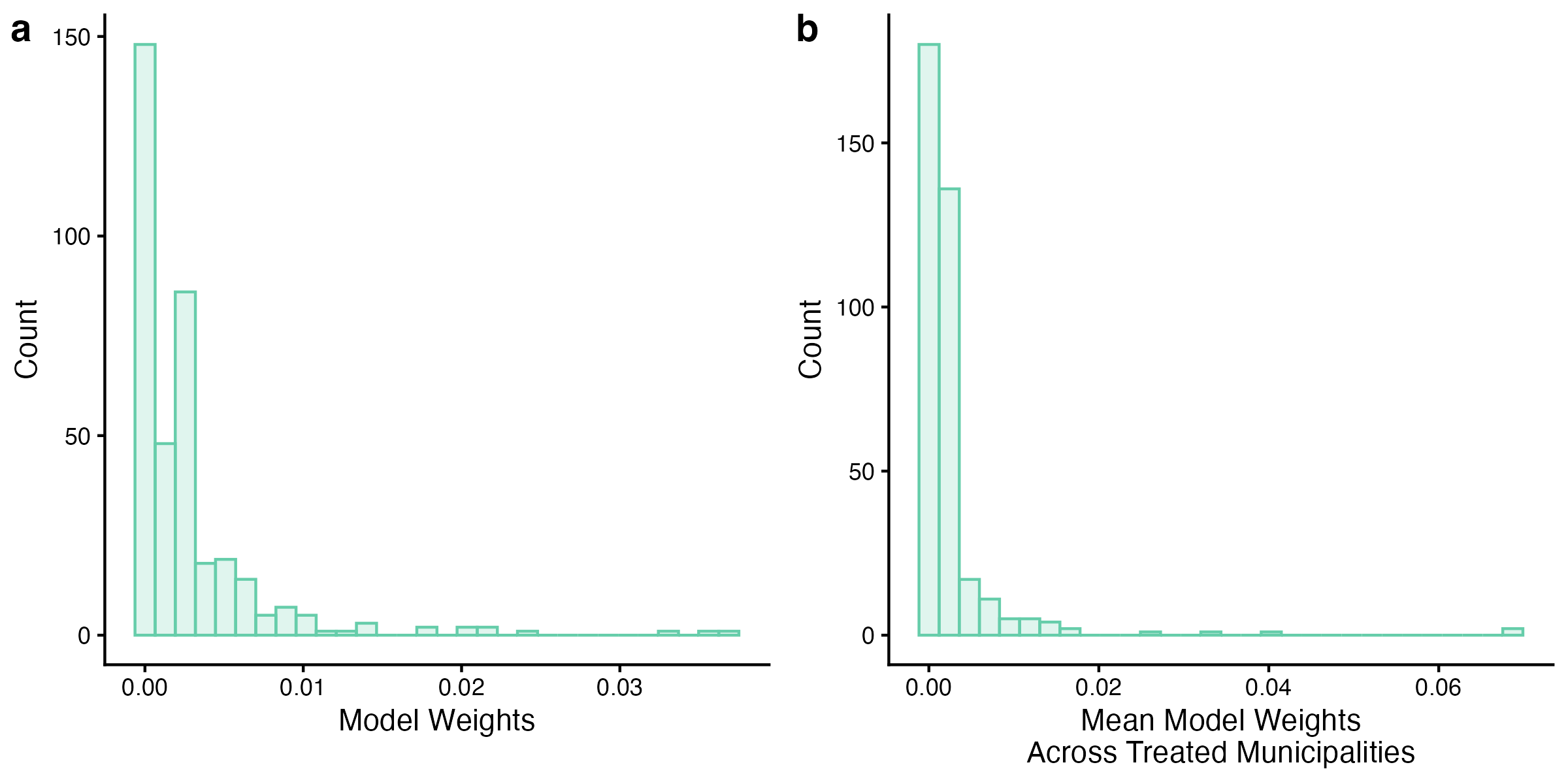}
    \caption{Model weights of augmented synthetic controls for the main analysis (\textbf{K = 5}): (a) treated municipalities pooled; (b) treated municipalities modeled separately. In (b), values are mean weights across treatment municipalities for visual clarity.}
    \label{fig:augsynth_weights}
\end{figure}

\end{document}